\documentclass[twocolumn]{autart}    
\usepackage[T1]{fontenc}
\usepackage[utf8]{inputenc}
\usepackage{lmodern}
\usepackage{graphicx}
\usepackage{amsmath,amssymb,amsfonts}
\let\theoremstyle\relax
\usepackage{amsthm}
\usepackage{algorithm}
\usepackage[noend]{algpseudocode}
\usepackage{bm}
\usepackage[hidelinks,bookmarks=false]{hyperref}

\makeatletter
\@ifundefined{theHpart}
  {}
  {}

\@ifundefined{theHsection}
  {}
  {}

\@ifundefined{theHsubsection}
  {}
  {}

\@ifundefined{theHsubsubsection}
  {}
  {}
\makeatother
\usepackage{textcomp}
\usepackage{xcolor}
\usepackage[authoryear,round]{natbib}

\allowdisplaybreaks[1]         
\usepackage{booktabs}
\usepackage{caption}
\usepackage[utf8]{inputenc}

\newtheorem{proposition}{Proposition}
\newtheorem{lemma}{Lemma}
\newtheorem{theorem}{Theorem}

\newtheorem{remarktemp}{Remark}
\newenvironment{remark}{\begin{remarktemp}\normalfont}{\end{remarktemp}}

\newtheorem{definitiontemp}{Definition}
\newenvironment{definition}{\begin{definitiontemp}\normalfont}{\end{definitiontemp}}

\newtheorem{assumptiontemp}{Assumption}
\newenvironment{assumption}{\begin{assumptiontemp}\normalfont}{\end{assumptiontemp}}

\usepackage[nodisplayskipstretch]{setspace}

\begin{document}

\begin{frontmatter}

\title{Continuous Aggregative LQG Games with Delayed Discrete Observations\thanksref{footnoteinfo}}
\thanks[footnoteinfo]{Corresponding author F.~Rajabali.}
\author[poly,gerad]{Farid Rajabali}\ead{farid.rajabali@polymtl.ca}, 
\author[poly,gerad]{Roland Malham\'e}\ead{roland.malhame@polymtl.ca}, \author[poly2]{Sadegh Bolouki}\ead{sadegh.bolouki@polymtl.ca}
\address[poly]{Department of Electrical Engineering, Polytechnique Montr\'{e}al, Canada}
\address[gerad]{GERAD, Canada}
\address[poly2]{Department of Computer and Software Engineering, Polytechnique Montr\'{e}al, Canada}


\begin{abstract}
Mean field game equilibria are predicated on the assumption of immediate pairwise interactions within a population of homogeneous agents with asymptotically vanishing influence as population size increases. However, in many real-world cases, agents receive the information on population with a delay. In this paper, we characterize agent best responses under an information exchange structure whereby agents observe the empirical mean state only at discrete time instants with some delay. Sufficient conditions are presented for the existence of a Nash equilibrium within a finite population of agents, and the change in cost due to discrete empirical mean observations with  delay  versus without delay and  global continuous state observations is also evaluated. 
\end{abstract}

\end{frontmatter}


\section{Introduction}
\label{sec:introduction}
Information underpins decision-making in multi-agent systems, particularly in contexts like Mean Field Games (MFGs) and aggregative games, where agents tailor their control policies based on population statistics. However, real-time access to such information is often challenging, in contrast to the perfect-information assumption used in much of the MFG literature \citep[e.g., see][] {huang2019linear, lacker2022case, lasry2007mean, huang2007large}.\\
Recent advances extend MFG theory to settings with partial observability. In a linear--quadratic--Gaussian (LQG) framework, for example, \cite{bensoussan2021linear} derive decentralized strategies under both individual and common noise using Kalman filtering and fixed‐point methods to establish $\varepsilon$-Nash equilibria, while \cite{huang2006distributed} similarly approximate the empirical mean from noisy local observations to obtain control policies. Specific structures of interaction  have also been studied  under partial observations (PO). Major–minor frameworks  which model systems with one influential major agent and a population of interacting minor agents  under various combinations of PO are considered for LQG cases in \cite{caines2016epsilon} and \cite{firoozi2020epsilon}. In nonlinear settings, i.e. characterized by nonlinear system dynamics or non-quadratic costs,  \cite{sen2019mean} employ nonlinear filtering and the separation principle to study the case of agents with PO of their own states, while \cite{sen2016mean} extend the results of \cite{caines2016epsilon} (minor agents having PO of the major agent) to the nonlinear situation. A discrete-time formulation with risk-sensitive criteria is provided by \cite{saldi2023partially}, who recast the problem on the belief space via dynamic programming. The Pontryagin stochastic maximum principle for mean-field control problems with PO is established in \cite{buckdahn2017mean}, while \cite{bensoussan2021mean} generalize the separation principle under non-Gaussian conditions.\\ 
 In this paper, we consider  an aggregative continuous-time  multi-agent game where agents' cost depends on the population empirical mean state, the latter being observed only at a sequence of discrete time points. This results in a  mixed discrete-continuous time stochastic game for which  we investigate the existence of Nash equilibria. Previous  work in this direction has been recently reported in \cite{rajabali2024partial}. However, while in \cite{rajabali2024partial} the  observation instants coincide with the time of the reported empirical mean, we consider here the more delicate case of time-delayed population mean observations.\\
While our analysis here is limited to a finite population case, our long term  objective is to explore the possibility of the onset of mean field equilibria (infinite populations)  despite delayed and possibly distorted observations of aggregate population statistics.  Indeed, one way to understand  fish school dynamics when the school is viewed as an interacting multi-agent system, is that despite  limited   observations of dynamically changing neighbors, individual fish appear to manage to form through sheer accumulation of data  a time delayed version of the group behavior \citep{rosenthal2015revealing}.  
\\
More generally, agents often know only their neighbors’ states—think pedestrian flow or traffic—while their performance/costs depend on a global population mean they cannot observe directly. Consequently, despite a situation of competition, the agents could find it advantageous to collaborate through  consensus schemes \citep[for instance,][]{olfati2004consensus,olfati2007consensus}  to turn appropriate local exchanges into delayed estimates of the unknown mean. Indeed, such collaboration  implicitly occurs when using navigation apps.  Our model is an idealized abstraction of that situation. It corresponds to the finite horizon, finite population LQG dynamic game solved in \cite{huang2019linear} except that where the global population state was continuously observed by all agents, it is only its empirical mean  which is now observed at discrete time instants and further revealed after a time delay. Under a number of assumptions of which the most important  is that individual agents choose to suppress any private information if, when accounted for, the symmetry of the game is broken,   we develop a dynamic programming (DP) analysis allowing computation of best responses and leading to an existence theorem for a Nash equilibrium.\\
Note that recent works such as \cite{WANG2023110789} and \cite{LIANG2024111518} also construct exact Nash equilibria for finite populations. However, except for assumed \textit{a priori} statistical information about the initial states (as is also the case in MFG analyses), the empirical mean is considered totally unobservable. \textit{Its impact on agents' costs is computed in an open-loop manner} by exploiting symmetries and the smoothing property of conditional expectations given strictly local information. In contrast, by capitalizing here on periodic, albeit delayed, observations of the empirical mean (potentially simulating the time required for a consensus-based exchange of information), agents achieve a distinct \textit{piecewise open-loop} notion of Nash equilibrium.\\
\textbf{Contributions.}  Under assumptions which will be gradually specified, the paper contributions can be summarized as follows:
\begin{enumerate}
    \item[(i)]  \textit{The model}.  At the modeling level, the paper introduces a new class of aggregative games with states evolving in continuous time, yet with delayed aggregate observations occurring in discrete time. We believe this model is a useful abstraction for analyzing collective dynamics over large networks which are sparsely connected, yet with connections  virtually made denser over time through  sharing of information.
    \item[(ii)] \textit{DP and predictive control analysis}. In terms of mathematical analysis, we face a piecewise open loop stochastic game where the control problem bears resemblance to an MFG predictive control situation. Application of DP is made difficult as the latter proceeds backwards while the time delay feature requires predictions over time intervals where control policy is not yet explicit. An assumption that the prediction of empirical mean at the start of each interval is a sufficient statistic, i.e. that it encapsulates all useful past information, allows us to express control policies in terms of that predicted quantity. The latter is treated as a virtual measurement and it  can only be evaluated in a forward manner starting from time 0, once control policies over all discrete time observation intervals have been determined via DP. The sufficient statistic assumption is ultimately shown to be  at least consistent with the (linear) control policy structure it leads to. 
    \item[(iii)]  \textit{Existence of a Nash equilibrium.} Sufficient conditions for the existence of a Nash equilibrium are specified.
    \item[(iv)] \textit{Performance analysis}. We quantify the impact of delay through a closed-form performance-loss formula relative to the full state observation version of the game in \cite{huang2019linear}  and corroborate the theory with simulations across delays and population sizes.
\end{enumerate}
The paper is organized as follows: Section 2 defines the game and costs. Sections 3 and 4 discuss the DP equations and derive the Nash Equilibrium. Section 5 evaluates the performance loss. Section 6 presents numerical results. Finally, Section 7 concludes the paper.

\section{Definition of the Game}
We consider a non-cooperative game with $N$ identical agents whose scalar dynamics are governed by
\begin{equation}\label{1}
dx_i(t) = \big(a x_i(t) + b u_i(t)\big) dt + \sigma dw_i(t), \quad t \geq 0
\end{equation}
where $x_i(t)$ is the state of agent
$i$, and $u_i(t)$ is its control input, and $a \in \mathbb{R},\ b\neq0,\ \sigma\geq 0$.  The $w_i(t),\ 1\le i \le N$, are independent Wiener processes with zero mean, and are independent of the agents' initial states (assumed to have finite variance). Each agent continuously observes its own state but receives the empirical mean
\begin{equation*}
\bar{x}^N(t_j) = \frac{1}{N} \sum_{i=1}^{N} x_i(t_j)
\end{equation*}
only at a later discrete time $t_{j+1} =  (j+1)\Delta t$, for $j=0, \dots, n-1$, where $n = \lfloor\frac{T}{\Delta t}\rfloor$, and $T$ is the time horizon. Thus, the information about the mean at time $t_j$ is received with a delay of $\Delta t$.
Agents aim to track a target trajectory given by $\Gamma \bar{x}^N(t)+\eta$\footnote{
The analysis extends to an affine target $\Gamma \bar{x}^N(t) + \eta$ by applying the state transformation $\tilde{x}_i = x_i - \eta/(1-\Gamma)$ (assuming $\Gamma \ne 1$). The transformation introduces a deterministic input into the state dynamics (a constant drift term), but it does not change the structure of the underlying feedback control problem. For notational simplicity, we proceed with $\eta = 0$.}\footnote{Possible values for $\Gamma$ are discussed in Appendix C.}, while simultaneously minimizing their control effort, as quantified by a cost function parameterized by $q, r > 0$ and $h \geq 0$:
\begin{equation}\label{2}
\begin{aligned}
J_i 
&= \mathbb{E}\!\Biggl[ 
     \int_{0}^{T} 
       \Bigl(q\,(x_i(t) - \Gamma\bar{x}^N(t))^{2} + r u_i^{2}(t)\Bigr)\,dt \\
&\quad + h\,\bigl(x_i(T) - \Gamma \bar{x}^N(T)\bigr)^{2} 
   \Biggr]
\end{aligned}
\end{equation}
The interaction between agents in this model occurs through the cost function, a standard formulation in aggregative games.  As to the choice of linear dependency $\Gamma \bar{x}^N(t)$ of  error on the tracked mean trajectory, values of $\Gamma=1$ or $\Gamma=-1$ are typical in crowd following or congestion avoidance applications respectively.

%
\section{Predictor-Driven DP Equations}
\begin{definition}
A strategy profile $u^*=(u_1^*,\ldots,u_N^*)$ is a Nash equilibrium if, for every agent $i$,
\begin{equation}\label{eq:NE}
J_i(u_i^*,u_{-i}^*) \le J_i(u_i,u_{-i}^*) \quad \forall\, u_i \in \mathcal{M}_i .
\end{equation}
Here $u_{-i}^*=(u_1^*,\ldots,u_{i-1}^*,u_{i+1}^*,\ldots,u_N^*)$, and\\
$
\mathcal{M}_i := \Big\{ u_i:\mathbb{E}\!\int_0^T u_i(t)^2\,dt<\infty$, $u_i(t)$ depends only on $\{x_i(s)\!:\,s\le t\}$ and $\{\bar{x}^N(t_k)\!:\, t_k+\Delta t \le t\}\Big\}$.
\end{definition}
Note that agents observe both their local state in continuous time and the empirical mean with delay at discrete sampling times. However, for reasons stated in Remark 1 below, when developing their interval-wise empirical mean predictor, agents \textit{choose to ignore} the local state information. Thus, in effect, for $t\in [t_j,t_{j+1})$, agents work with a ``partially degraded'' empirical mean predictor defined as follows:
\begin{equation}\label{5}
\begin{aligned}
\hat{\bar{x}}_{j-1}^N(t) 
&:= \mathbb{E}\!\left[\bar{x}^N(t) \,\Big|\, \bar{x}^N(t_{j-1})\right], 
\end{aligned}
\end{equation}
 Note  $j-1$ in \eqref{5} refers to the most recent available observation $\bar{x}^N(t_{j-1})$. Furthermore, we define the prediction error for the time interval $t \in [t_j, t_{j+1})$ as:
\begin{equation}\label{7} 
\Delta_{j-1}(t) = \bar{x}^N(t) - \hat{\bar{x}}_{j-1}^N(t).
\end{equation}
The assumptions underlying this paper are outlined in the following.
\begin{assumption}[Agents preserve the symmetry of the game].\ Agents choose to suppress any private information that would disrupt the symmetry of the game. In particular, they aim to achieve   interval-wise common estimates of the empirical mean.\label{assump1}
\begin{remark}
  If agents were to do away with Assumption 1, and given that they continuously observe their own state (but not other agents'  states), they would then have to account for the fact that the (common) empirical mean observations relay \textit{different information to different agents}. Accounting for this  would break the symmetry of the game and  lead to an infinite regress of reciprocal beliefs situation, characteristic of games of incomplete information \citep{Harsanyi1967PartI}. Agents are thus assumed to avoid this conundrum by working  under the constraint that any information that would lead to differentiation in their empirical mean predictions, in particular due to accounting for   the \textit{specifics} of their own state in the estimator, is ignored. Under such a  constraint, this leads in effect to  interval-wise \textit{deterministic empirical mean predictors}, because interval prediction must now occur in the absence of new information on the empirical mean other than its latest (shared by all agents) observation.
\end{remark}
\end{assumption}
\begin{assumption}[Empirical mean predictors are an interval-wise  Markov process]\label{assump2}
In essence it is assumed that the empirical mean predictor given all observations up to time $t_j$ is,  in and of itself, a sufficient statistic for empirical mean predictions until the next observation becomes available, i.e. on the complete interval $[t_j,t_{j+1})$.
\begin{remark} Assumption 2 is a strong assumption. While it does hold for linear control policies,  it will not  in general hold if the control policies that it leads to turn out to be nonlinear. Thus, the only argument we can offer in support of the assumption is one of \textit{plausibility}; namely by  verifying that, adopting it,  leads to best response policies whose structure remains consistent with the assumption.  In particular, it is the case here as the resulting policies will be shown to be affine feedback laws in the agent state.
\end{remark}
\begin{remark} Let the
predictor state at $t_j$ be $\hat{\bar x}^{N}_{j-1}(t_j)$. This predictor must be conditional on the latest empirical mean observation, i.e. that at time $t_{j-1}$.  However, a subtle point is that in a DP framework proceeding backwards, the best response law on $[t_{j-1}, t_j]$ is not yet known at time $t_j$. Thus the required predictor cannot be computed even if the empirical mean at $t_{j-1}$  is assumed available at $t_j$. Yet, as we proceed with DP, we shall treat this predictor as if it were available as a so-called ``virtual measurement''. Interval-wise control laws will be expressed in terms of such virtual measurements. The latter become computable in a forward direction, starting from knowledge of the empirical mean predictor at the start of the control horizon (because at that early stage, the empirical mean predictor is assumed as an a priori common guess by all agents). Note that this notion of virtual measurement is quite different from the virtual measurements of Kalman filtering. The latter are synthetic observations often added based on reasonable physically-based guesses to address issues of insufficient observability in state estimation problems.
\end{remark}
\end{assumption}
\begin{definition}
We define $u_i$ as a Markov strategy if it is expressed as a function of the current time, the agent’s current state, and the most recent virtual measurement:
\begin{equation}
u_i(t) = f\!\left(t, x_i(t), \hat{\bar{x}}_{j-1}^N(t)\right),\quad t\in[t_j,t_{j+1}).
\end{equation}
\end{definition}
Note that the use of the predictor $\hat{\bar{x}}_{j-1}^N ( t_{j})$ as an argument above is justified by Assumption 2.
To find the best response control $u_i^*$, we begin by introducing the value function for each agent $i = 1,\dots,N$, and for each time interval $j = 0,\dots,n-1$ with $t \in [t_j, T]$, as follows:
\begin{multline}\label{4}
V_{i}(t,X_{i,j-1})
= \inf_{u_i \in \mathcal{M}_i} \mathbb{E} \Bigl[ 
\int_{t}^{T} \Bigl( q\!\left( x_i(\tau) - \Gamma \bar{x}^N(\tau) \right)^{2} \\
+ r u_i^{2}(\tau) \Bigr) d\tau 
+ h\!\left( x_i(T) - \Gamma \bar{x}^N(T) \right)^{2} 
\,\Big|\, X_{i,j-1}(t) \Bigr] \\
= \inf_{u_i \in \mathcal{M}_i} \mathbb{E} \Bigl[
\int_{t}^{t_{j+1}} \Bigl( q\!\left( x_i(\tau) - \Gamma \bar{x}^N(\tau) \right)^{2} 
+ r u_i^{2}(\tau) \Bigr) d\tau 
 \\ + V_{i}(t_{j+1}, X_{i,j})
\,\Big|\, X_{i,j-1}(t) \Bigr]
\end{multline}
where $X_{i,j-1}(t)$ denotes the pair $(x_i(t), \bar{x}^{N}(t_{j-1}) )$. Since agents do not have continuous access to $\bar{x}^N(t)$,  we find it convenient to express the running cost in the value function above in terms of a perturbation of its most recent predictor, i.e. $\hat{\bar{x}}^N_{j-1}(t)$.   This will lead us in the following lemma to the conclusion that the best response policies can be computed based on an adjusted cost function for $t \in [t_j, T)$ where $\bar{x}^N(t)$ is now replaced by its most current predictor. The associated  ``adjusted''\ family of value functions is: 
\begin{gather} 
\widetilde{V}_{i}(t, \widehat X_{i,j-1}) = \inf_{u_i \in \mathcal{M}_i} \mathbb{E} \Bigl[ \int_{t}^{T} \Bigl( q \left( x_i(\tau) - \Gamma \hat{\bar{x}}_{j-1}^N(\tau) \right)^2   \nonumber \\ +r u_i^2(\tau) \Bigr) d\tau 
+ h \left( x_i(T) - \Gamma \hat{\bar{x}}_{n-1}^N(T) \right)^2 \Big| \widehat {X}_{i,j-1}(t) \Bigr].
\label{8}
\end{gather}
where $\widehat X_{i,j-1}(t)$ denotes the pair $(x_i(t), \hat{\bar x}^{N}_{j-1}(t))$. In the following, we shall omit writing the time argument, $t$, whenever no ambiguity arises.

\begin{lemma}
The best response policy associated with the adjusted cost function $\widetilde{V}_{i}(t, \widehat{X}_{i,j-1})$, is identical to that associated with the original cost function $V_{i}(t, {X}_{i,j-1})$.
\end{lemma}
\begin{proof}
To reformulate the cost minimization, we express $\bar{x}^N(\tau)$ in terms of its predicted value $\hat{\bar{x}}_{j-1}^N(\tau)$ and the prediction error $\Delta_{j-1}(\tau)$. Also, we write the DP equation for $V_{i}$ and will show that control policies are identical for both cost functions interval-wise. Thus we rewrite the original cost-to-go for $[t_j,t_{j+1})$ in \eqref{4} as follows:
\begin{multline}\label{9}
V_{i}(t,  X_{i,j-1})
=\inf_{u_i \in \mathcal{M}_i} \mathbb{E}\Bigl[ \int_{t}^{t_{j+1}} 
\Bigl( q\!\left( x_i(\tau) - \Gamma \hat{\bar{x}}_{j-1}^N(\tau)\right)^{2} \\
+ q\Gamma^{2} \bigl(\Delta_{j-1}(\tau)\bigr)^{2} 
+ 2q\Gamma \Delta_{j-1}(\tau) \bigl( x_i(\tau) - \Gamma \hat{\bar{x}}_{j-1}^N(\tau)\bigr) \\
+ r u_i^{2}(\tau) \Bigr) d\tau 
+ V_{i}(t_{j+1}, X_{i,j})  \,\Big|\, X_{i,j-1}(t) \Bigr].
\end{multline}
As one can see in \eqref{9}, expanding the quadratic term in the original cost-to-go \eqref{4} yields terms dependent on $(x_i - \Gamma\hat{\bar{x}}_{j-1}^N)$, on $\Delta_{j-1}$ and a cross-term. The term dependent on $(x_i - \Gamma\hat{\bar{x}}_{j-1}^N)$ corresponds to the adjusted cost $\widetilde{V}_{i}$.  Since by Remark 1, all agents share a common deterministic empirical mean predictor which shapes their best responses, and the stochastic character of the prediction error is only affected by internal noises independent of the controls, the term  $(\Delta_{j-1})^2$ will not affect the form of the best response. By Assumption 1, agents choose to suppress knowledge of the local state when carrying out the interval-wise prediction of the empirical mean. As a result, in their cost calculations, they replace $\mathbb{E}[\bar{x}^N(\tau)|X_{i,j-1}(t)]$ with the estimator $\hat{\bar{x}}_{j-1}^N(\tau)$ previously defined in (4). Consequently, the conditional expectation of the prediction error becomes:
\begin{equation}
  \mathbb E\!\Bigl[ \Delta_{j-1}(\tau)\,\Big|\, \bar {x}^N(t_{j-1})\Bigr]= 0.
  \label{10-3}
\end{equation}
Indeed as written in  \eqref{10-3}, with this redefined view of information the empirical mean estimator $\hat{\bar{x}}_{j-1}^N(\tau)$ becomes unbiased, and it is orthogonal to the estimation error $\Delta_{j-1}(\tau)$ (by the orthogonality property of conditional expectation, see Chapter 3 of \cite{luenberger1997optimization}). Thus:
\[
\mathbb E\!\Big[\Delta_{j-1}(\tau)\,\hat{\bar x}_{j-1}^N(\tau)\,\Big|\, \bar{x}^{N}(t_{j-1})\Big] = 0.
\]
Furthermore, in view of \eqref{10-3}:
\begin{gather*}
\mathbb{E}\!\Bigl[
\Delta_{j-1}(\tau)\,x_i(\tau)\,\Big|\, X_{i,j-1}(t)
\Bigr]=0.
\end{gather*}
As a result, since the difference $V_i - \widetilde{V}_i$ is a term that $u_i$ cannot influence, the optimal policy that minimizes $V_i$ also minimizes $\widetilde{V}_i$. This argument will hold for all intervals by backward induction, starting from the terminal cost expression.
\end{proof}
\begin{remark}  In summary, for estimation purposes of the empirical mean  on time intervals $[t_j,t_{j+1})$, agents rely only on observation $\bar{x}^{N}(t_{j-1})$. In the subsequent DP developments, and by virtue of Assumption 2, information $\bar{x}^{N}(t_{j-1})$ on such time intervals  will be replaced by the sufficient statistic $\hat{\bar x}^{N}_{j-1}(t)$. \end{remark}
\begin{remark}  Except for suppressing knowledge of local state  $x_i$ in evaluating its own predictor of the empirical mean, the agent \textit{still makes full use of the information on its own state} in building its best response, as this remains compatible with the symmetry of the game. 
\end{remark}
Given the discrete nature of observations, while the control must evolve in continuous time, deriving via DP the response policy for agent $i$ requires  solving both a discrete step dynamic program accounting for discrete observations, and an interval-wise continuous dynamic program embedded within the discrete program in between successive observations as outlined below.\\
1. \textit{Embedded interval-wise DP equation}: The DP equation for $\widetilde{V}_{i}(t_j,\widehat{X}_{i,j-1})$ is evaluated at $t\in [t_{j},t_{j+1})$, $j=1,\ldots,n-2$, as follows:
\begin{multline}\label{11}
\widetilde{V}_{i}(t, \widehat{X}_{i,j-1}) = \inf_{u_i \in \mathcal{M}_i} \mathbb{E} \Bigl[ \int_{t}^{t_{j+1}} \Bigl( q\big(x_i(\tau) - \Gamma \hat{\bar{x}}_{j-1}^N(\tau)\big)^2    \\ 
+r u_i^2(\tau) \Bigr) d\tau  + \widetilde{V}_{i}(t_{j+1}, \widehat{X}_{i,j}) \,\Big|\,\widehat {X}_{i,j-1}(t) \Bigr].
\end{multline}
Between consecutive sampling times, the Hamilton-Jacobi-Bellman (HJB) equation is solved as a tracking problem where the predicted empirical mean $\hat{\bar{x}}_{j-1}^N(t)$ is treated as a deterministic but unknown function. The boundary conditions established at $t_{j+1}$ serve as the terminal conditions for solving the interval-wise HJB equations. This approach yields the best response policy for agent $i$ in terms of its state $x_i$ and the predictor.

2. \textit{Discrete component of the DP equation}: 
Boundary conditions (BCs) at the terminal time $T$ and at intermediate times $t_{j+1}$, for $j=1,\ldots,n-1$ are given by:
\begin{equation}\label{12}
\widetilde{V}_{i}(T,\widehat{X}_{i,n-1})
= h\!\left(x_i(T) - \Gamma \hat{\bar{x}}_{n-1}^N(T)\right)^{2}
\end{equation}
\begin{multline}\label{13}
\widetilde{V}_{i}(t_{j+1},\widehat{X}_{i,j-1})
=\mathbb{E}\!\left[\widetilde{V}_{i}(t_{j+1},\widehat{X}_{i,j})
\,\middle|\,\widehat  {X}_{i,j-1}(t_{j+1})\right].
\end{multline}
\section{Dynamic Programming Equation Analysis Over Time Intervals}
\subsection{Control Policy Computation for $[t_{n-1}, T)$}
In this section, we solve DP equation \eqref{11} interval-wise and start from the last interval $[t_{n-1}, T)$ as the boundary conditions impose a backwards propagating solution. A deterministic trajectory $\hat{\bar{x}}_{n-2}^N(t)$, defined for $t \in [t_{n-1},T)$, is assumed. According to Lemma 1, determining the best response policy requires solving DP equation \eqref{11}. To achieve this, we derive the HJB equation for $\widetilde{V}_{i}(t,\widehat{X}_{i,n-2})$, which will yield the optimal control policy over the interval $[t_{n-1},T)$.
\begin{equation}\label{14}
\begin{aligned}
0 &= \frac{\partial \widetilde{V}_{i}}{\partial t} 
   + \min_{u_i} \Biggl[ 
      \frac{\partial \widetilde{V}_{i}}{\partial x_i}\,(a x_i + b u_i) \\
&\quad + \Bigl(q \left(x_i - \Gamma \hat{\bar{x}}_{n-2}^N\right)^{2} + r u_i^{2}\Bigr) 
   + \frac{1}{2} \sigma^{2}\,\frac{\partial^{2} \widetilde{V}_{i}}{\partial x_i^{2}}
   \Biggr]
\end{aligned}
\end{equation}
We define the variables $p(t)$, $s(t)$, and $\rho{}(t)$ to represent the solution of the HJB equation associated with $\widetilde{V}_{i}(t, \widehat{X}_{i,n-2})$, where the value function is assumed to take the following quadratic form:
\begin{equation}
\widetilde{V}_{i}(t, \widehat{X}_{i,n-2}) = p(t) x_i^2(t) + 2s_{n-2}(t)x_i(t) + \rho_{n-2}(t)
\label{15}
\end{equation}
In view of the Hamiltonian in \eqref{14}, and given the quadratic form assumed in \eqref{15}, the optimal control policy $u_i^*$ can be expressed as:
\begin{equation}\label{16}
\begin{aligned}
u_i^*(t) 
= -\frac{b}{2r}\,
   \left( \tfrac{\partial \widetilde{V}_i(t,\widehat{X}_{i,n-2})}{\partial x_i} \right)
= -\frac{b}{r}\,\bigl( p(t)x_i(t) + s_{n-2}(t) \bigr).
\end{aligned}
\end{equation}
To solve the HJB in \eqref{14}, we need to specify the boundary condition for $\widetilde{V}_{i}(t, \widehat{X}_{i,n-2})$ at $t=T$: 
\begin{multline}
\widetilde{V}_{i}(T, \widehat{X}_{i,n-2})  =\mathbb{E} \Bigl[ \widetilde{V}_{i}(T,\widehat{X}_{i,n-1})    \Big| \widehat{X}_{i,n-2}(T) \Bigr]   \\=h x_i^2(T) - 2h\Gamma x_i(T)\mathbb{E}\left[\hat{\bar{x}}_{n-1}^N(T)\Big|\hat{\bar{x}}_{n-2}^N(T)\right] \\ +h\Gamma^2\mathbb{E}\left[(\hat{\bar{x}}_{n-1}^N(T))^2\Big|\hat{\bar{x}}_{n-2}^N(T)\right].
\end{multline}
In order to solve HJB in \eqref{14}, we substitute \eqref{15} and \eqref{16} in it and solve for finding differential equations for $p(t)$, $s(t)$ and $\rho(t)$ as follows: 
\begin{equation}\label{18}
\frac{dp}{dt} = -2pa + \frac{b^2}{r}p^2 - q, 
\qquad p(T) = h
\end{equation}
\begin{equation}\label{19}
\begin{aligned}
\frac{ds_{n-2}}{dt} 
&= -\left(a - \frac{b^2}{r} p \right) s_{n-2} 
   + q \Gamma \hat{\bar{x}}_{n-2}^N, \\
s_{n-2}(T) &= -h \Gamma \hat{\bar{x}}_{n-2}^N(T)
\end{aligned}
\end{equation}
\begin{equation}\label{20}
\begin{aligned}
\frac{d \rho_{n-2}}{dt} 
&= \frac{b^2}{r} s_{n-2}^2 
   - q \bigl(\Gamma \hat{\bar{x}}_{n-2}^N\bigr)^2 
   - \sigma^2 p, \\
\rho_{n-2}(T) 
&= h \Gamma^2 \,\mathbb{E}\!\left[(\hat{\bar{x}}_{n-1}^N(T))^2 \,\middle|\, \hat{\bar{x}}_{n-2}^N(T)\right].
\end{aligned}
\end{equation}

\begin{assumption}\label{assump3}
The following Riccati differential equation has a unique solution on $[0,T]$:
\begin{gather}
\frac{d \alpha(t)}{dt} 
= -2 \left( a - \frac{b^2}{r} p(t) \right) \alpha(t) 
   + \frac{b^2}{r}\,\alpha^2(t) + q \Gamma, \nonumber\\
\alpha(T) = -h \Gamma.
\label{21}
\end{gather}
\end{assumption}
In the following proposition, we identify the nature of the Nash equilibrium of the game on $[t_{n-1},T]$ and establish  the predictor equation.
\begin{proposition}
If Assumptions \ref{assump1}-\ref{assump3} hold, then, for $t \in [t_{n-1}, T)$, the set of Markov Nash equilibrium strategies for the agents $i = 1, \ldots, N$, is given by:
\begin{gather}
u_i^*(t) = -\frac{b}{r} \left( p(t) x_i(t) + \alpha(t) \hat{\bar{x}}_{n-2}^N(t) \right)
\label{22}
\end{gather}
where the evolution of the predicted empirical mean state, $\hat{\bar{x}}_{n-2}^N(t)$, is described by the following differential equation with initial condition of virtual measurement, $\hat{\bar{x}}^N_{n-2}(t_{n-1})$:
\begin{equation}\label{23}
\frac{d \hat{\bar{x}}_{n-2}^N(t)}{dt} = \left( a - \frac{b^2}{r} p(t) - \frac{b^2}{r} \alpha(t) \right) \hat{\bar{x}}_{n-2}^N(t)
\end{equation}
\end{proposition}
\begin{proof}
Based on Definition 1, the $u_i^*$ in \eqref{16} which minimizes $\widetilde{V}_{i}$ is the Nash equilibrium of the game since any deviation from this policy by any agent would fail to minimize its cost, assuming that other agents have frozen their strategies. 
Given that $\hat{\bar{x}}_{n-2}^N(\tau)$ is a deterministic, although initially unknown trajectory, the optimal control problem in \eqref{11} resembles a standard LQG tracking problem. However, since all agents perform the same analysis due to homogeneity, a common feedback control policy must emerge. For consistency, this policy must be such  that, when applied by all agents, the resulting closed-loop predictor of the empirical population mean matches the trajectory being tracked. This requirement is analogous to finding equilibrium control policies in MFGs, where a fixed-point condition is essential for consistency. In this context, predictor functions are updated each time a new global empirical mean is observed, necessitating the use of DP equations on an interval-wise basis. In order to calculate this fixed-point, and following the approach in \cite{malhame2020mean}, we assume the structure of $s_{n-2}(t)$ to be of the form: 
\begin{equation}\label{24}
s_{n-2}(t) = \alpha(t) \hat{\bar{x}}_{n-2}^N(t)
\end{equation}
This assumption enables the decoupling of the forward and backward components of the overall solution. In order to express $\alpha(t)$ and $\hat{\bar{x}}^N_{n-2}$, one can substitute $u_i^*(t)$ from \eqref{16} in \eqref{1} to find the closed-loop control equation for $x_i(t)$. By averaging this equation over all agents and imposing the fixed-point consistency property, the differential equation for the empirical mean $\bar{x}^N(t)$ and subsequently, that of  the predictor  $\hat{\bar{x}}^N_{n-2}(t)$, equation \eqref{23}, can be determined. Now substituting \eqref{24} into \eqref{19} and recalling \eqref{23} as well as the boundary condition for \eqref{19}, one obtains \eqref{21} for $\alpha(t)$.
\end{proof}
\begin{remark}\label{remark6}
With the predictor differential equation established, the predictor can now be expressed as a function of the virtual measurement $\hat{\bar{x}}^N_{n-2}(t_{n-1})$. The predictor evolves with transition kernel $\phi_p(t,s):=\exp\!\bigl[\int_s^t \bigl(a-\frac{b^{2}}{r}(p(\tau)+\alpha(\tau))\bigr)\,d\tau\bigr]$. Similarly, the open-loop agent dynamics evolve with kernel $\phi(t,s):=\exp\!\bigl[\int_s^t \bigl(a-\frac{b^{2}}{r}p(\tau)\bigr)\,d\tau\bigr]$. In accordance with Assumption~\ref{assump2}, we have:
\begin{equation}\label{25}
\hat{\bar{x}}^N_{n-2}(t)=\phi_p(t,t_{n-1})\hat{\bar{x}}^N_{n-2}(t_{n-1})
\end{equation}
Note the propagation and update equations provide an explicit update mechanism, thereby verifying Assumption~\ref{assump2}.

Once $\hat{\bar{x}}^N_{n-2}(t)$ is determined, $s_{n-2}(t)$ becomes known, allowing the differential equation \eqref{20} to be solved for $\rho_{n-2}(t)$. Furthermore, $\rho_{n-2}(t)$ can be expressed in the form $\psi_{n-1}(t)\hat{\bar{x}}^N_{n-2}(t_{n-1})^2 + \gamma_{n-1}(t)$ for $t\in[t_{n-1},T)$ as follows:
\begin{multline}\label{27}
\rho_{n-2}(t)= -\int_{t}^{T} \Bigl( \frac{b^2}{r}s_{n-2}^2(\tau) - q\left(\Gamma \hat{\bar{x}}_{n-2}^N(\tau)\right)^2
 \\ - \sigma^2 p(\tau) \Bigr) d\tau 
+ h\Gamma^2\mathbb{E}\left[(\hat{\bar{x}}_{n-1}^N(T))^2\Big|\hat{\bar{x}}_{n-2}^N(T)\right] 
\end{multline}
The value of \  $\mathbb{E}[(\hat{\bar{x}}_{n-1}^N(T))^2|\hat{\bar{x}}_{n-2}^N(T)]$ is calculated in Appendix B,
with $\psi_{n-1}(t)$ and $\gamma_{n-1}(t)$ for $t\in [t_{n-1},T)$  given by:
\begin{equation*}
\begin{aligned}
\psi_{n-1}(t) 
&= -\int_{t}^{T} \phi_p^2(\tau, t_{n-1})
      \left( \frac{b^2}{r}\,\alpha^2(\tau) - q\Gamma^2 \right)\,d\tau \\
&\quad + h\Gamma^2 \phi_p^2(T, t_{n-1})
\end{aligned}
\end{equation*}
\begin{equation*}
\gamma_{n-1}(t) 
= \int_{t}^{T}\sigma^2 p(\tau)\,d\tau 
   + h\Gamma^2 \frac{\sigma^2}{N}
     \int_{t_{n-2}}^{t_{n-1}} \phi^2(T,s)\,ds
\end{equation*}
\end{remark}
\subsection{Control Policy on a generic interval $[t_j,t_{j+1})$}
In this section, we derive the best response policy over generic interval $[t_j, t_{j+1})$ by solving the HJB equation for $\widetilde{V}_{i}$ to generalize the results obtained in Proposition 1. For this purpose, proceeding recursively, we assume that the following expression holds for $\widetilde{V}_{i}$ for $ t \in [t_{j+1}, t_{j+2})$:
\begin{equation}\label{28}
\begin{aligned}
\widetilde{V}_{i}(t,\widehat{X}_{i,j})
&= p(t)\,x_i^{2}(t) 
   + 2\alpha(t)\,x_i(t)\,\hat{\bar{x}}^N_{j}(t) \\
&\quad + \psi_{j+1}(t)\,\bigl(\hat{\bar{x}}^N_{j}(t)\bigr)^{2} 
   + \gamma_{j+1}(t)
\end{aligned}
\end{equation}
where $\psi_{j+1}$ and $\gamma_{j+1}$ are assumed known time functions. Also, we assume that an equation analogous to \eqref{23} has been established for $\hat{\bar{x}}^N_{j}(t)$. Referring to DP equation in \eqref{11} and the boundary condition in \eqref{13}, the HJB equation for $\widetilde{V}_{i}$ and for $ t \in [t_{j}, t_{j+1})$ can be written as:
\begin{equation}\label{14.2}
\begin{aligned}
0 &= \frac{\partial \widetilde{V}_{i}}{\partial t} 
   + \min_{u_i} \Biggl[ 
        \frac{\partial \widetilde{V}_{i}}{\partial x_i}\,(a x_i + b u_i) \\
&\quad + \Bigl(q\,(x_i - \Gamma \hat{\bar{x}}_{j-1}^N)^{2} + r u_i^{2}\Bigr) 
      + \frac{1}{2}\sigma^{2}\,\frac{\partial^{2} \widetilde{V}_{i}}{\partial x_i^{2}}
   \Biggr]
\end{aligned}
\end{equation}
 To proceed further, we assume the following quadratic form for $\widetilde{V}_{i}$: 
\begin{equation}\label{29}
\widetilde{V}_{i}(t,\widehat{X}_{i,j-1})=p(t)x_i^2(t)+2s_{j-1}(t)x_i(t)+\rho_{j-1}(t)
\end{equation}
where $s_{j-1}(t)$ and $\rho_{j-1}(t)$ are coefficients to be determined, provided the assumed expression actually holds.
By differentiating $\widetilde{V}_{i}$ with respect to $x_i$, the optimal control policy $u_i^*=- \tfrac{b}{r} ( p(t) x_i(t) + s_{j-1}(t) )$ is obtained and substituted into the HJB equation. Solving the HJB equation, subject to the boundary condition in \eqref{13}, and substituting the quadratic form in \eqref{29} into \eqref{14.2}, leads to a system of differential equations for $p(t)$, $s_{j-1}(t)$, and $\rho_{j-1}(t)$, analogous to \eqref{18}, \eqref{19}, and \eqref{20}, respectively. Specifically, recalling boundary conditions, the equations for $s_{j-1}(t)$ and $\rho_{j-1}(t)$ can be written as:
\begin{equation}\label{31}
\begin{aligned}
\frac{ds_{j-1}}{dt} 
&= -\left(a - \frac{b^2}{r} p \right) s_{j-1} 
   + q \Gamma \hat{\bar{x}}_{j-1}^N, \\
s_{j-1}(t_{j+1}) 
&= \alpha(t_{j+1})\,
   \mathbb{E}\!\left[\hat{\bar{x}}^N_{j}(t_{j+1}) 
   \,\middle|\, \hat{\bar{x}}^N_{j-1}(t_{j+1})\right] \\
&= \alpha(t_{j+1})\,\hat{\bar{x}}^N_{j-1}(t_{j+1})
\end{aligned}
\end{equation}
\begin{equation}\label{32}
\begin{aligned}
\frac{d \rho_{j-1}}{dt} 
&= \frac{b^2}{r} s_{j-1}^2 
   - q \left(\Gamma \hat{\bar{x}}_{j-1}^N\right)^{2} 
   - \sigma^2 p, \\
\rho_{j-1}(t_{j+1}) 
&= \mathbb{E}\!\Bigl[ 
       \psi_{j+1}(t_{j+1})\,\bigl(\hat{\bar{x}}^N_{j}(t_{j+1})\bigr)^{2} 
       \\&+ \gamma_{j+1}(t_{j+1}) 
       \,\Big|\, \hat{\bar{x}}^N_{j-1}(t_{j+1})
   \Bigr]
\end{aligned}
\end{equation}
As in the construction leading to \eqref{27}, $\rho_{j-1}(t)$ admits a quadratic representation whose coefficients are obtained by backward propagation; we make this explicit in \eqref{34}--\eqref{37}. To solve for $s_{j-1}(t)$, we assume, inspired by \eqref{24} and the boundary condition in \eqref{31}, that:
\begin{equation}\label{extra1} 
s_{j-1}(t) = \alpha(t)\hat{\bar{x}}^N_{j-1}(t). 
\end{equation}
The resulting best response policy for the interval $[t_j, t_{j+1})$ is then:
\begin{equation}\label{33}
u_i^*(t) =  - \frac{b}{r} \left( p(t) x_i(t) + \alpha(t)\hat{\bar{x}}^N_{j-1}(t) \right).
\end{equation}
Note that the DP/Riccati solution yields \eqref{33}
so the  best response is Markov in $(x_i(t),\hat{\bar x}^{N}_{j-1}(t))$.
Next, we use the fixed-point requirement to derive the evolution of the virtual measurement 
$\hat{\bar{x}}_{j-1}^N(t)$ over the interval $[t_j, t_{j+1})$. By substituting the optimal control 
$u_i^*(t)$ from \eqref{33} into the individual state dynamics and averaging states of all agents, differential equation for $\hat{\bar{x}}_{j-1}^N(t)$ similar to \eqref{23} will be  obtained. Also, note that $p(t)$  in \eqref{29} and $\alpha(t)$ in \eqref{extra1} will satisfy Riccati equations \textit{respectively consistent} with  \eqref{18} and \eqref{21}.
Finally, $\rho_{j-1}(t)$ is computed by solving \eqref{32} over $t \in [t_j, t_{j+1})$, as follows:
\begin{multline}
{\rho}_{j-1}(t)= -\int_{t}^{t_{j+1}} \Bigl( \frac{b^2}{r} s_{j-1}^2(\tau) - q(\Gamma \hat{\bar{x}}_{j-1}^N(\tau))^2
- \\ \sigma^2 p(\tau) \Bigr) d\tau 
 + \rho_{j-1}(t_{j+1})
=\psi_j(t)\hat{\bar{x}}^N_{j-1}(t_{j})^2+\gamma_j(t)
\label{34}
\end{multline}
where $\psi_j(t)$ and $\gamma_j(t)$ are computed recursively:
\begin{equation}\label{36}
\begin{aligned}
\psi_j(t) 
&= \psi_{j+1}(t_{j+1})\,\phi_{p}(t_{j+1}, t_{j})^{2} \\
&\quad - \int_{t}^{t_{j+1}} 
   \phi_{p}(\tau, t_{j})^{2}\,
   \left( \frac{b^2}{r}\,\alpha^{2}(\tau) - q\Gamma^{2} \right) 
   d\tau
\end{aligned}
\end{equation}
\begin{equation}\label{37}
\begin{aligned}
\gamma_j(t) 
&= \gamma_{j+1}(t_{j+1}) 
   + \sigma^{2} \int_{t}^{t_{j+1}} p(\tau)\,d\tau \\
&\quad + \psi_{j+1}(t_{j+1})\,\frac{\sigma^{2}}{N} 
   \int_{t_{j-1}}^{t_{j}} \phi^{2}(t_{j+1},s)\,ds.
\end{aligned}
\end{equation}
The value of $\mathbb{E}\left[\hat{\bar{x}}_{j}^N(t_{j+1})^2 \, \middle| \, \hat{\bar{x}}_{j-1}^N(t_{j+1})\right]$ is detailed in Appendix B. These results allow us to recursively compute $\psi_j(t)$ and $\gamma_j(t)$ for interval $[t_j, t_{j+1})$.
\subsection{Virtual Measurement}
Thus far, following Assumption~\ref{assump2}, we have derived the best-response policy and the predictor equation, assuming knowledge of the virtual measurement $\hat{\bar{x}}^N_{j-1}(t_j)$. In this section, we specify this quantity and propose a structure for it. To achieve this, we close the loop in \eqref{1} by incorporating the best response policy derived in \eqref{33}. The detailed steps for these calculations are provided in Appendix A, with the resulting expression for $[t_j, t_{j+1})$ given as:
\begin{equation}\label{45_2}
\begin{aligned}
\hat{\bar{x}}^N_{j-1}(t_j) 
&= \phi(t_j,t_{j-1})\,\bar{x}^N(t_{j-1}) \\
& - \frac{b^2}{r} 
   \int_{t_{j-1}}^{t_j} 
      \phi(t_j,s)\,\alpha(s)\,\phi_{p}(s,t_{j-1})\,
      \hat{\bar{x}}_{j-2}^N(t_{j-1})\,ds.
\end{aligned}
\end{equation}
In particular, $\hat{\bar{x}}^N_{j-1}(t_j)$ summarizes all past reveals, so the predictor on $[t_j,t_{j+1})$ depends on the past only through this quantity, as posited in Assumption~\ref{assump2}.\\
DP solution for $[t_j,t_{j+1})$ and finding the virtual measurement lead us to the following theorem which is the main result of the paper.
\begin{theorem}
If Assumptions \ref{assump1}-\ref{assump3} hold, then, the set of Markov Nash equilibrium strategies for $k=1,\ldots,N$ are given by:
\begin{equation}\label{42}
u_k^*(t) = -\frac{b}{r} \left( p(t) x_k(t) + \alpha(t) \hat{\bar{x}}_{j-1}^N(t) \right)
\end{equation}
where the predictor $\hat{\bar{x}}_{j-1}^N(t)$ satisfies the following equation:
\begin{equation}\label{43}
\hat{\bar{x}}^N_{j-1}(t)=\phi_{p}(t,t_{j})\hat{\bar{x}}^N_{j-1}(t_{j})
\end{equation}
and $\hat{\bar{x}}^N_{j-1}(t_{j})$ is found via \eqref{45_2}.
\end{theorem}
\begin{remark} 
Note that under Assumption~1, the Nash equilibrium is exact for finite $N$, not an approximation. In line with standard MFG theory, we expect that as $N \to \infty$, the empirical mean becomes deterministic, which would naturally eliminate the need for Assumption~1 to prevent an infinite recursion of beliefs. However, rigorously establishing $\epsilon$-Nash bounds under this specific delayed and discrete observation structure requires a formal limit argument, which we defer to future work.
\end{remark}
\begin{remark}
The uniqueness of solution requirement in Assumption \ref{assump3} is imposed so that no ambiguity arises as to what best response agents should opt for. Please refer to Appendix C for a set of sufficient conditions on the parameters of the cost and dynamics, or the length of the control horizon $[0,T]$, so as  to prevent finite escape time in the Riccati equation of Assumption \ref{assump3} above. By continuity and boundedness, existence and uniqueness of solutions would then be guaranteed.

\end{remark}
\begin{remark}
 Recalling Remark 3 above, the virtual measurements (Empirical mean predictors at start of intervals $[t_j,t_{j+1})$ cannot be computed starting at an arbitrary time $t_j, j \ge 2$ even if the observation of empirical mean at time $t_{j-1}$ is available. The only way to carry that evaluation is to start by computing the virtual measurement at $t=0$ which is indeed the only point where it is arbitrarily initialized via a common a priori guess by all agents. Subsequently, one proceeds forward in time, to generate future virtual measurements sequentially as empirical mean observations become available and based on  \eqref{45_2}. For simplicity, we set $\hat{\bar{x}}_{-1}^N(0) = 0$.
\end{remark}
In Algorithm 1, we summarize the steps required to compute the Nash equilibrium of the game as outlined in Theorem 1.
\begin{algorithm}[t]
\caption{DP Algorithm for Best Response Computation}
\label{alg:DP-NE}
\begin{algorithmic}

\State \textbf{Phase I (Backward Pass):} Solve the Riccati differential equations on $[0, T]$ for $p(t)$ in \eqref{18} and $\alpha(t)$ in \eqref{21}.

\State \textbf{Phase II (Forward Pass):} Initialization:
\Statex \hspace{2em}
\begin{equation*}
    \hat{\bar{x}}_{-1}^N(0) = 0
\end{equation*}

\For{$j = 0,\ldots,n-1$}

\Statex \hspace{1em} \textbf{1) Predictor Evolution for $\bm{t \in [t_j,\,t_{j+1})}$:}
\Statex
    \begin{equation*}
        \hat{\bar{x}}_{\,j-1}^N(t)
      \;=\;
        \phi_{p}\!\bigl(t,\,t_j\bigr)\;
        \hat{\bar{x}}_{\,j-1}^N\!\bigl(t_j\bigr)
    \end{equation*}

\Statex \hspace{1em}
\textbf{2) Nash equilibrium Control Computation for $t \in [t_j, t_{j+1})$:}
\Statex \hspace{2em} For each agent $k=1, \dots, N$, the control is:
\Statex
    \begin{equation*}
        u_k^*(t)
      \;=\;
      -\frac{b}{r}\,\Bigl(
        p(t)\,x_k(t)
        \;+\;\alpha(t)\,\hat{\bar{x}}_{\,j-1}^N(t)
      \Bigr)
    \end{equation*}

\Statex \hspace{1em} \textbf{3) Virtual Measurement Update at $\bm{t_{j+1}}$:} Given $\bar{x}^N(t_{j})$ and the known $p(\cdot)$ and $\alpha(\cdot)$ from Phase I, compute $\hat{\bar{x}}_{j}^N(t_{j+1})$ using \eqref{45_2}.

\EndFor

\end{algorithmic}
\end{algorithm}
\subsection{Prediction Error Formula}
To determine the prediction error for $t \in [t_j, t_{j+1})$, $j=1,\ldots,n-1$, the necessary steps are outlined at the end of  Appendix A, where the empirical mean is expressed in terms of the predictor, resulting in:
\begin{equation}\label{44}
\Delta_{j-1}(t) = \sigma \int_{t_{j-1}}^{t} \phi(t,s) d\bar{w}^N(s)
\end{equation}
where $d\bar{w}^N(s)=1/N\sum_{k=1}^N dw_k(s)$ and $\phi$ is the open-loop kernel defined in Remark \ref{remark6}. The error in \eqref{44} is consistent with Lemma 1, as the difference reduces to a term depending only on noise through the prediction error variance.
\begin{remark} The model can be extended to the multidimensional case, where the scalars $(a,b,\Gamma, q,r,h)$ become matrices
$(A,B,\Gamma,Q,R,H)$. While existence and uniqueness of solutions to the multidimensional equivalent of (17) can be subsumed by imposing conditions guaranteeing convexity of the cost, conditions sufficient for Assumption~3 to hold are unfortunately not within reach. The DP steps remain identical.\end{remark}
\section{Performance Evaluation}
In this section, the value of the original cost function $V_{i}(t, X_{i,j-1})$ is derived from the value of the adjusted cost function $\widetilde{V}_{i}(t,\widehat X_{i,j-1})$. Since the empirical mean observation is updated every $\Delta t$ seconds, we \emph{decompose} the difference $V_{i}-\widetilde{V}_{i}$ over successive intervals via Bellman’s recursion. Leveraging the DP equations in \eqref{9} and \eqref{11} and the proof of Lemma~1, for $t\in[t_j,t_{j+1})$ and $j=1,\ldots,n-1$ we obtain
\begin{multline}\label{42.2}
\Delta V_{j-1}(t)  := {V}_{i}(t, X_{i,j-1})- \widetilde{V}_{i}(t,\widehat{X}_{i,j-1})  \\
=  \mathbb{E} \left[ \int_t^{t_{j+1}} \left(q \left(x_i(\tau) - \Gamma \bar{x}^N(\tau)\right)^2 - q\left(x_i(\tau) \right. \right. \right. \\ \left. \left.  \left. -\Gamma \hat{\bar{x}}_{j-1}^N(\tau)\right)^2 \right) d\tau
+ \Delta V_{j}(t_{j+1}) \middle| \widehat X_{i,j-1} \right]  \\= q\Gamma^2\mathbb{E} \left[ \int_t^{t_{j+1}}  \Delta_{j-1}^2(\tau) d\tau  \right] + \mathbb{E}\left[\Delta V_{j}(t_{j+1})\middle|\widehat X_{i,j-1}\right].
\end{multline}
At time $t=T$ the performance difference is
\begin{equation}
\mathbb{E}\left[\Delta V_{n-1}(T)\middle|\widehat X_{i,n-1}\right]=h\Gamma^2\mathbb{E}\left[\Delta_{n-1}^2(T)\right].
\end{equation}
Furthermore, in \eqref{42.2}, 
\begin{equation*}
\mathbb{E} \left[ \int_{t_j}^{t_{j+1}}  \Delta_{j-1}^2(\tau) d\tau  \right]=\frac{\sigma^2}{N}  \int_{t_j}^{t_{j+1}} \int_{t_{j-1}}^{\tau}\phi^2(\tau,s) dsd\tau.
\end{equation*}
Finally, we define the total performance loss for $[0,T]$ for generic agent $i$ relative to the case of \textit{continuous global state observations} in \cite{huang2019linear} as
\begin{equation}
\Delta J_i \;:=\; V_i(0, {X}_{i,-1}) \;-\; V_i^{Cont}(0,x_i,\bar{x}^N)
\label{DeltaJ}
\end{equation}
with  the expression of  $V_i^{Cont}$ in \cite{huang2019linear}  recalled below as:
\begin{multline}\label{46.2}
V_i^{Cont}(0, x_i, \bar{x}^N) = p(0)x_i^2(0) + 2\,\alpha(0)\bar{x}^N(0)x_i(0) \\+ \psi(0)(\bar{x}^N(0))^2 + \sigma^2 \int_{0}^{T} p(\tau)\,d\tau.
\end{multline}
Using \eqref{29} and \eqref{42.2} one can compute $V_i(0, X_{i,-1})$. Further using \eqref{46.2} into \eqref{DeltaJ}, $\Delta J_i$ is obtained.
\section{Simulations}
In addition to comparison with the performance result in \cite{huang2019linear}, a further performance comparison with the case of discrete observations  with \textit{zero-latency} of the empirical mean (see \cite{rajabali2024partial}) is reported. Parameters are as in Fig.~\ref{fig1}.\\
 Figure~\ref{fig1} compares the three total costs as functions of the sampling period $\Delta t$ and the population size $N$. For a moderate population ($N=10$), the delayed case with $\Delta t=2$\,s is the most expensive, reflecting both discrete sampling and the added reveal delay. For fixed $\Delta t$, increasing $N$ drives the discrete–information costs toward the continuous–information cost, consistent with the $O(1/N)$ variance scaling of the empirical mean.
\begin{figure}[!t]
\centerline{\includegraphics[width=\columnwidth]{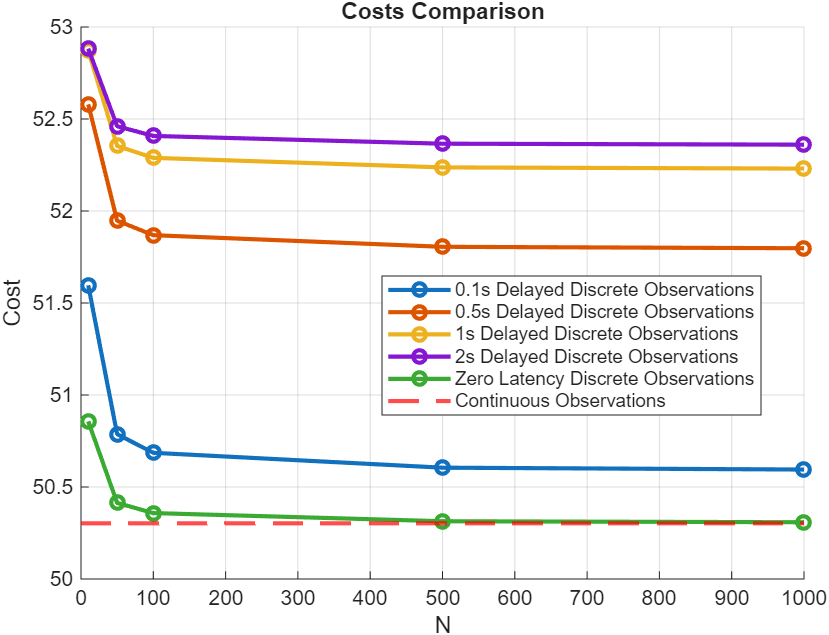}}
\caption{Cost comparison across information patterns for varying $\Delta t$ and $N$.
Parameters: $a=1$, $b=1$, $q=1$, $r=1$, $T=20$ s, $\sigma=1$, $\Gamma=0.8$.}
\label{fig1}
\end{figure}\\
Figure~\ref{fig2}  reports the \emph{total delay penalty}
$\Delta J$ as a function of $\Delta t$ for $N=100$.
The curve is monotone and gently saturating, in agreement with what we expect when delay increases.
\begin{figure}[!t]
  \centering
  \includegraphics[width=\columnwidth]{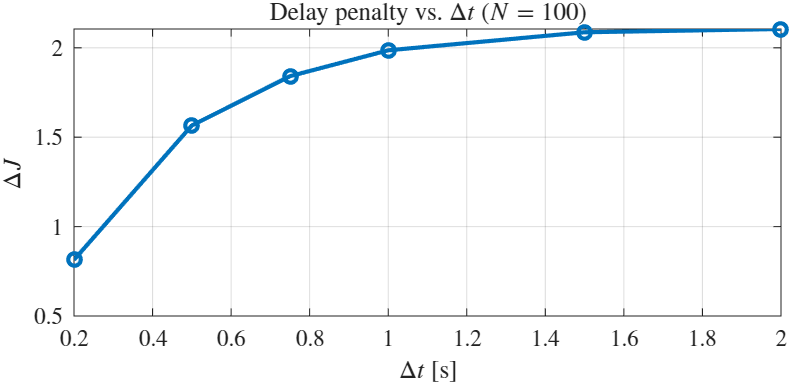}
  \caption{Delay penalty $\Delta J$ vs.\ sampling period $\Delta t$ for $N=100$ and $T=20$ s}.
  \label{fig2}
\end{figure}
Figure~\ref{fig3} shows the ensemble mean $\bar{x}^{N}(t)$ (solid) and the delayed predictor
$\hat{\bar{x}}^{N}_{j-1}(t)$ (dashed) for two sampling periods. Vertical dotted lines mark reveal times $t_{j+1}$.
Between reveals, $\hat{\bar x}^{N}_{j-1}(t)$ evolves deterministically according to the predictor dynamics; at each reveal it will be updated based on the newly available sample.
With the ``zero'' initial prior, the predictor is flat on $[0,t_1)$ and then snaps to $\bar x^N(0)$ at the first reveal.
Larger $\Delta t$ produces larger tracking error (visible band widening and slower predictor corrections), which matches the penalty growth in Fig.~\ref{fig2}.
\begin{figure}[!t]
  \centering
 \includegraphics[width=\columnwidth]{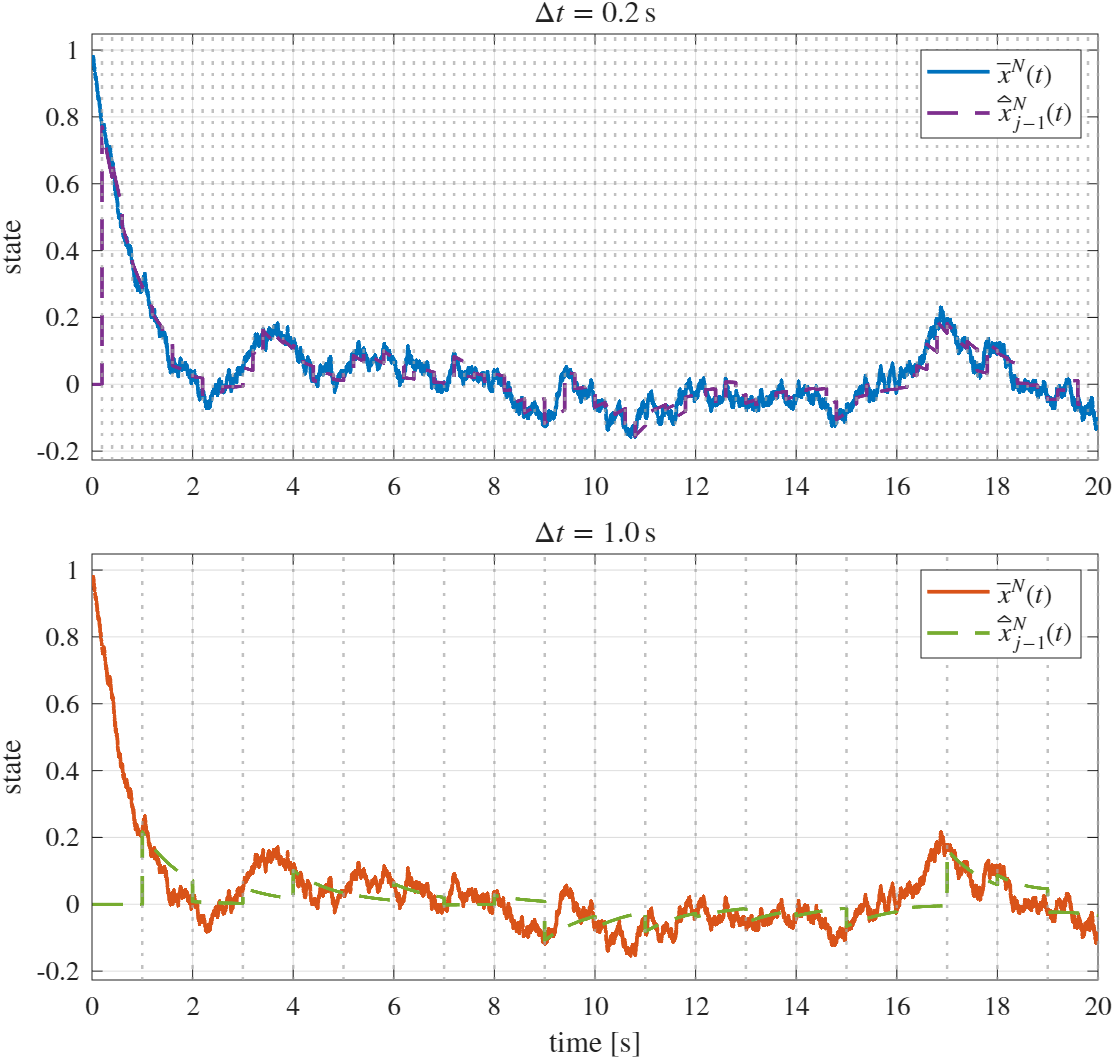}
  \caption{Empirical mean $\bar{x}^{N}(t)$ (solid) and delayed predictor
  $\hat{\bar{x}}^{N}_{j-1}(t)$ (dashed).($N=100$). Vertical dotted lines: reveal times $t_{j+1}$. Top: $\Delta t=0.2$\,s. Bottom: $\Delta t=1.0$\,s.}
  \label{fig3}
\end{figure}
\section{Conclusion}
Our research has explored  aggregative games featuring LQG dynamics with delayed discrete observations of aggregate state. By introducing the concept of virtual measurement as a sufficient statistic of past observations, we successfully developed Markovian Nash strategies using DP principles. Numerical evaluations of the loss of performance induced by delays in empirical mean  observations were carried out. The results appear to be generalizable to multidimensional systems and multiple observation time delay steps. More importantly though, they are likely to constitute an analytical  building block for the analysis of aggregative games on networks where agents learn about empirical mean via consensus algorithms.
\section*{Appendices}
\subsection*{A. Calculation of Virtual Measurement and Empirical mean}
Under the best response policy  $u_i^*$ in \eqref{33}, and for $t \in [t_{j-1}, t_{j})$ , $x_i(t)$ evolves as:
\begin{multline}\label{380}
dx_i(t) = \left(a x_i(t) + b u_i(t)\right) dt + \sigma dw_i(t) =\Bigl(a x_i(t) \\ -\frac{b^2}{r} (p(t) x_i(t) + \alpha(t) \hat{\bar{x}}_{j-2}^{N}(t))\Bigr) dt + \sigma dw_i(t) 
\end{multline}
Integrating \eqref{380} from $t_{j-1}$ to $t \in [t_{j-1}, t_j)$ and
averaging over all $i = 1, \dots, N$ results in:
\begin{gather}\label{em}
\bar{x}^N(t) = \phi(t,t_{j-1}) \bar{x}^N(t_{j-1}) - \\ \nonumber\frac{b^2}{r} \int_{t_{j-1}}^{t} \phi(t,s) \alpha(s) \hat{\bar{x}}_{j-2}^N(s) ds + \sigma \int_{t_{j-1}}^{t} \phi(t,s) d\bar{w}^N(s). 
\end{gather}
At $t=t_j$, the empirical mean $\bar{x}^N(t_{j-1})$ 
is revealed for the first time. Using \eqref{5}, \eqref{em}, and Lemma 3 of \citet{rajabali2023can}, one obtains the \emph{new} virtual measurement as follows:
\begin{gather}
\hat{\bar{x}}^N_{j-1}(t_j)= \mathbb{E}\Bigl[\bar{x}^N(t_j)\Big|\bar{x}^N(t_{j-1})\Bigr]= \phi(t_j,t_{j-1}) \bar{x}^N(t_{j-1}) \nonumber \\- \frac{b^2}{r} \int_{t_{j-1}}^{t_j} \phi(t_j,s) \alpha(s) \phi_{p}(s,t_{j-1}) \hat{\bar{x}}_{j-2}^N(t_{j-1}) ds.
\label{41}
\end{gather}
Subtracting \eqref{em} and \eqref{41} at time $t_j$ yields:
\begin{equation}\label{46}
\bar{x}^N(t_{j}) = \hat{\bar{x}}^N_{j-1}(t_{j})+  \sigma \int_{t_{j-1}}^{t_{j}} \phi(t_{j},s) d\bar{w}^N(s).
\end{equation}
%

\subsection*{B. Calculation of  $\mathbb{E}[\hat{\bar{x}}^N_{j}(t_{j+1})^2|\hat{\bar{x}}^N_{j-1}(t_{j+1})]$}
First,
$\hat{\bar{x}}^N_{j}(t_{j+1})$ is computed from \eqref{41} and $\bar{x}^N(t_{j})$ from \eqref{46}:
\begin{gather}
\hat{\bar{x}}^N_{j}(t_{j+1})=\phi(t_{j+1},t_{j}) \bar{x}^N(t_{j}) \nonumber \\- \frac{b^2}{r} \int_{t_{j}}^{t_{j+1}} \phi(t_{j+1},s) \alpha(s) \phi_{p}(s,t_{j}) \hat{\bar{x}}_{j-1}^N(t_{j}) ds=\nonumber \\ \phi(t_{j+1},t_{j})\hat{\bar{x}}^N_{j-1}(t_{j}) + \phi(t_{j+1},t_{j}) \sigma \int_{t_{j-1}}^{t_{j}} \phi(t_{j},s) d\bar{w}^N(s) +\nonumber \\- \frac{b^2}{r} \int_{t_{j}}^{t_{j+1}} \phi(t_{j+1},s) \alpha(s) \phi_{p}(s,t_{j}) \hat{\bar{x}}_{j-1}^N(t_{j}) ds = \nonumber \\\hat{\bar{x}}^N_{j-1}(t_{j+1}) + \phi(t_{j+1},t_{j}) \sigma \int_{t_{j-1}}^{t_{j}} \phi(t_{j},s) d\bar{w}^N(s)
\end{gather}
then straightforward calculations yield:
\begin{multline}\label{65}
\mathbb{E}[\hat{\bar{x}}^N_{j}(t_{j+1})^2|\hat{\bar{x}}^N_{j-1}(t_{j+1})] = \hat{\bar{x}}^N_{j-1}(t_{j+1})^2+ \\ \frac{\sigma^2}{N} \int_{t_{j-1}}^{t_{j}} \phi^2(t_{j+1},s) ds.
\end{multline}
\subsection*{C. On the Existence and Boundedness of $\alpha(t)$}
This appendix provides the conditions under which the solution to the Riccati differential equation for $\alpha(t)$ is guaranteed to exist and remain bounded on the finite horizon $[0, T]$.
We analyze the Riccati equation for the sum $\beta(t) := p(t) + \alpha(t)$.
\begin{equation}
\dot{\beta} = \frac{b^2}{r}\beta^2 - 2a\beta - q(1 - \Gamma), \quad \beta(T) = h(1 - \Gamma).
\label{eq:app_beta}
\end{equation}
To analyze this equation for finite-escape singularities, we perform a standard transformation. Let $w(t) = \frac{b^2}{r}\beta(t) - a$. Then $w(t)$ satisfies the canonical Riccati equation:
\begin{equation}
\dot{w} = w^2 + E,
\label{eq:app_w}
\end{equation}
where the constant $E$ is given by:
\begin{equation}
E := \frac{b^2}{r}q(\Gamma - 1) - a^2.
\label{eq:app_E}
\end{equation}
The behavior of the solution to \eqref{eq:app_w} depends critically on the sign of $E$.\\
\textbf{Case 1: $E \le 0$.}
If $E \le 0$, the solution to \eqref{eq:app_w} is of the form of a hyperbolic tangent, which is bounded for all finite time. Therefore, no finite-escape singularity can occur. The condition $E \le 0$ is equivalent to:
$$ \frac{b^2}{r}q(\Gamma - 1) - a^2 \le 0 \quad \iff \quad \Gamma \le 1 + \frac{a^2 r}{b^2 q}. $$
Under this condition, $\beta(t)$ remains bounded on $[0, T]$. Since $p(t)$ is already known to be bounded, it follows that $\alpha(t) = \beta(t) - p(t)$ must also be bounded.\\
\textbf{Case 2: $E > 0$.}
If $E > 0$ (i.e., $\Gamma > 1 + \frac{a^2 r}{b^2 q}$), the solution to \eqref{eq:app_w} involves a tangent function, which admits finite-escape singularities. Solving backwards from $w(T)$, the escape time $t_{esc}$ is given by:
$$ t_{esc} = T - \frac{\pi/2 - \arctan(w(T)/\sqrt{E})}{\sqrt{E}}. $$
A singularity occurs within the horizon $[t_0, T]$ if $t_{esc} > t_0$. For any fixed horizon length $T-t_0$, this condition can always be met by choosing a sufficiently large $\Gamma$. To guarantee a well-defined Nash equilibrium on any arbitrary finite horizon, we must preclude this possibility. The condition $E \le 0$, imposed via the bound on $\Gamma$ in Assumption 3, is a simple and sufficient condition to ensure no such singularity occurs.

\bibliographystyle{elsarticle-harv}
\bibliography{autosam}
\end{document}